\documentclass[12pt]{article}

\usepackage{amsfonts}
\usepackage{amssymb}
\usepackage{amsmath}
\usepackage{amsthm}
\usepackage{eurosym}
\usepackage{tikz}
\usepackage{graphicx}
\usepackage{multirow}
\usepackage{subfigure}
\usepackage{pgfplots}
\RequirePackage[colorlinks,citecolor=blue,linkcolor=red,urlcolor=blue]{hyperref}
\RequirePackage{natbib}

\usetikzlibrary{shapes,snakes}

\theoremstyle{theorem}
\newtheorem{theorem}{Theorem}

\newtheorem{exam}{Example}
\newenvironment{example}{\begin{exam} \rm }{\hfill  $\triangleleft$ \end{exam}}
\newtheorem*{exam1}{Example 1 (continued)}

\newtheorem*{exam3}{Example 3 (continued)}

\newtheorem{assum}{Assumption}

\newtheorem{rema}{Remark}
\newenvironment{remark}{\begin{rema} \rm }{\hfill $\triangleleft$ \end{rema}}
\newtheorem{defin}{Definition}

\newtheorem{mechan}{Mechanism}

\pgfplotsset{my style/.append style={axis x line=middle, axis y line=
           middle}}

\def\citeapos#1{\citeauthor{#1}'s (\citeyear{#1})}

\DeclareFontFamily{U}{mathx}{\hyphenchar\font45}
\DeclareFontShape{U}{mathx}{m}{n}{
      <5> <6> <7> <8> <9> <10>
      <10.95> <12> <14.4> <17.28> <20.74> <24.88>
      mathx10
      }{}
\DeclareSymbolFont{mathx}{U}{mathx}{m}{n}
\DeclareMathSymbol{\bigtimes}{1}{mathx}{"91}

\makeatletter

\renewcommand*{\@seccntformat}[1]{%
  \csname the#1\endcsname.\quad}
\makeatother

\begin{document}

\title{Identification of misreported beliefs}

\author{
\textsc{Elias Tsakas}\footnote{Department of Economics (MPE), Maastricht University, P.O. Box 616, 6200 MD, Maastricht, The Netherlands; Homepage: \url{www.elias-tsakas.com}; E-mail: \href{mailto:e.tsakas@maastrichtuniversity.nl}{\texttt{e.tsakas@maastrichtuniversity.nl}}}\\ 
\small{\textit{Maastricht University}}}

\date{\small{December 2021}}

\maketitle

\begin{abstract}

\noindent It is well-known that subjective beliefs cannot be identified with traditional choice data unless we impose the strong assumption that preferences are state-independent. This is seen as one of the biggest pitfalls of incentivized belief elicitation. The two common approaches are either to exogenously assume that preferences are state-independent, or to use intractable elicitation mechanisms that require an awful lot of hard-to-get non-traditional choice data. In this paper we use a third approach, introducing a novel methodology that retains the simplicity of standard elicitation mechanisms without imposing the awkward state-independence assumption. The cost is that instead of insisting on full identification of beliefs, we seek identification of misreporting. That is, we elicit beliefs with a standard simple elicitation mechanism, and then by means of a single additional observation we can tell whether the reported beliefs deviate from the actual beliefs, and if so, in which direction they do.

\vspace{0.5\baselineskip}



\end{abstract}

\section{Introduction}

Being able to obtain unbiased estimates of people's beliefs is of outmost importance for explaining and predicting behavior, and designing policy interventions \citep{Manski2004}. However, given the inherent latency of beliefs, obtaining such estimates has to rely heavily on (self-)reporting. Unfortunately, in practice, people often misreport their beliefs, even when they are incentivized to tell the truth. Thus, an important question is to identify whether their reports indeed deviate from their actual beliefs, and if so, in which direction.

As an illustration of the problem, take the usual finding: Democrats and Republicans systematically disagree --- usually by a large margin --- on the probability they reportedly assign to some politically charged event, e.g., the winner of the next elections \citep{Bullock2015}. The literature on politically motivated beliefs would say that the two groups actually have different beliefs. However, there is an alternative explanation: actual beliefs are not that divergent, and differences are amplified by misreporting (in the direction of their preferred parties respectively). Making this distinction is crucial, e.g., when deciding on whether to regulate (mis-)information. This is because, actual belief polarization --- as opposed to mere exaggeration  of reported beliefs --- has the potential to trigger extreme political reactions. 

Another common observation is that most people report that they are themselves more skilled than their average peer, e.g., when asked about their driving skills  \citep{Svenson1981}. Once again, the literature on motivated beliefs would say that this is consistent with the actual beliefs being biased in favor of the preferred state. But once again, one could argue that this is because people misreport their beliefs about their own perceived ability. And again, the distinction can have important consequences, e.g., for insurance purposes. If people actually overestimate their abilities, they may end up taking more risks when they drive, as opposed to situations where they just exaggerate their self-reported abilities.

But why would people misreport even when they are incentivized to tell the truth? A number of psychological factors have been recently proposed in the literature, such as self-image \citep{EwersZimmermann2015}, preference to appear truthful to their audience \citep{Thaler2021}, deliberate attempt to express attitudes, known as ``cheerleading" \citep{Bullock2015, HannonRidder2021}. The common thread among these explanations is that people have state-dependent preferences, i.e., they have some sort of stakes in the realization of the underlying event. 

Among theorists this is not particularly surprising: it is well known for years that, whenever preferences are state-dependent, beliefs cannot be identified using only data on traditional choices among acts \citep{Fishburn1973, KarniSchmeidlerVind1983, Dreze1987}. In particular, even if we somehow observed the complete preference relation (over acts), it would still be the case that for every belief there would exist a (state-dependent) utility function such that the resulting expected utility function would represent these preferences. So, as colossal as \citeapos{Savage1954} subjective expected utility theory is, it will only help us to identify beliefs if preferences are exogenously assumed to be state-independent.\footnote{The same holds true for the famous subsequent attempt of \cite{AnscombeAumann1963}.} And of course, the same problem is inherited by almost all belief elicitation mechanisms, viz., proper scoring rules \citep{Brier1950, Good1952,  Savage1971}, binarized scoring rules \citep{HossainOkui2013}, matching probabilities \citep{DucharmeDonnell1973, KadaneWinkler1988, Baillon2018}, clock auctions \citep{Karni2009, Tsakas2019}, promisory notes \citep{DeFinetti1974, KadaneWinkler1988}. This is because all these mechanisms essentially boil down to choosing from menus of acts, i.e., they rely on traditional choice data.
 
As a response to this conundrum, two approaches have been taken in the literature, which we will call the ``practically oriented" and the ``theoretically sound" approach. 

According to the practically oriented approach, it is still assumed, in the spirit of \cite{Savage1954} and \cite{AnscombeAumann1963}, that preferences are state-independent. The idea is pretty pragmatic: existing elicitation mechanisms are simple and easy to implement. This is what makes them appealing after all. Thus, we are keen to maintain this simplicity even if it comes at the price of imposing an exogenous structural assumption (viz., state-independence). 

On the other hand, the theoretically sound approach dispenses the awkward state-independence assumption. However, this means that we need to go well beyond traditional choice data, in ways that make belief elicitation practically intractable. In this sense, it is not surprising that this literature is quite thin \citep{Karni1999, JaffrayKarni1999} and has not been adopted by empirical researchers.\footnote{There is also large related literature within axiomatic decision theory. However, the different conditions that lead to identification of beliefs (and utilities) are perhaps even more demanding. Of course, from the point of view of this literature, this is not a concern, as the overall aim is to show that identification is in principle feasible and therefore the notion of subjective beliefs is well-defined. We further elaborate on the specific contributions in the literature section.}

It is not hard to guess that we find neither of the two approaches very satisfactory. So, we propose a third alternative, which maintains the simplicity of the existing elicitation tasks, without at the same time imposing the awkward state-independence assumption. The price that we will pay for reconciling the two is that we no longer insist on full identification of beliefs, but rather we only seek identification of misreporting. In other words, we will be able to tell if the reported beliefs --- that we have already elicited using our standard simple techniques --- deviate from the actual beliefs, and if so in which direction. For instance, in our earlier example, we will not be able to pin down the exact beliefs of each Democrat and each Republican, but we will be able to tell which ones exaggerate when they report their forecasts. Thus, we will be able to test whether the difference in their beliefs has been amplified by the fact that they both have stakes in the event they are forecasting.


Our method is on a high level inspired by the moral hazard literature \citep{Dreze1987, DrezeRustichini1999, Baccelli2021}, in that we exploit the presence of some action which is known to affect the agent's belief in a certain direction.\footnote{The term ``moral hazard" should not be confused with the one used in information economics. } We call such an action \textit{influential}. Of course, the crucial difference is that in the moral hazard literature the influential action is controlled by the agent herself, whereas in our case it is controlled by the analyst (Remark \ref{REM:moral hazard}). For instance, in the context of our earlier example, an influential action would be a donation to the campaign of the Democratic candidate, or the initiation of a negative rumor for the Republican candidate, in some swing state. In fact, any action that helps one of the two candidate would do the trick, even if the help is marginal (see Examples \ref{EX:investor}-\ref{EX:driver}). Although we cannot quantify the effect of each of these actions on an agent's beliefs, we can safely assume that there will be an increase in the probability of the Democratic candidate winning. 

Then, our method proceeds as follows. First, we elicit beliefs using a proper binarized scoring rule, i.e., an incentive-compatible scoring rule that pays in probabilities to win a fixed prize.\footnote{Later in the paper, we show that our analysis holds verbatim for any incentive-compatible binarized elicitation mechanism, e.g., binarized matching probabilities, or clock auctions (Section \ref{S:other elicitation methods}). The common feature of all these mechanisms is that they do not need to impose any assumption on the subject's risk preferences, which is what makes them very appealing. We also discuss the extension to non-binarized mechanisms, such as arbitrary proper scoring rules (Section \ref{S:non-binarized scoring rules}).} Subsequently, we ask the agent to choose between two fifty-fifty lotteries, whose outcomes are combinations of whether the prize is paid and whether the influential action is taken. In the context of the previous example, the two lotteries can be labeled as a ``risky option" and a ``hedging option". The risky option is a coin toss that either will pay both the prize to the agent and the donation to the campaign, or will not pay any of the two. On the other hand, the hedging option is a coin toss that either will only pay the prize to the agent, or will only pay the donation to the campaign. Then, our main result shows that the choice between the two lotteries identifies misreporting in the preceding belief elicitation task (Theorem \ref{T:binarized}). In particular, if the risky option (resp., the hedging option) is chosen, the reported probability of the Democratic candidate winning is greater (resp. smaller) than the actual belief. If the subject is indifferent between the two options, then the reported belief is the same as the actual belief, i.e., there is no misreporting.

The bottomline is that our approach allows us to keep using the state-of-the-art elicitation methodology, and only adds on top a simple task which identifies whether the agent has misreported beliefs.\footnote{In Section \ref{S:hedging}, we explain that distortions due to hedging opportunities are not really a concern.} Thus, simply put, we make a significant step in solving a long-standing problem (viz., belief elicitation under state-dependent preferences) at a very small cost (viz., adding a single question to the current methodology). Of course, our solution is partial, but in many applied settings --- where a qualitative analysis is used --- full identification is anyway not needed. Besides, if our method concludes that there is no misreporting, full identification is achieved. 

The only papers that introduce mechanisms for eliciting beliefs under state-dependent preferences are \cite{Karni1999} and \cite{JaffrayKarni1999}, with the latter proposing two different mechanisms. In particular, \cite{Karni1999} and the first mechanism of \cite{JaffrayKarni1999} rely on the assumption that state utilities are bounded, and they approximate the actual beliefs in the limit as monetary incentives grow arbitrarily large.\footnote{For an extensive discussion on the boundedness of the utility function, see \cite{Wakker1993}.} This is a rather uncomfortable convention, as the elicitation task will rely on a very large dataset. Moreover, we will either need to incur an extremely high cost, or to use hypothetical data. These problems are recognized by the authors of the two aforementioned papers, who point out that in those early days of the literature there was no other option \citep[e.g.,][p.485]{Karni1999}. The second mechanism in \cite{JaffrayKarni1999} assumes that state-dependence enters the picture in terms of unobserved state-dependent payments. So, first, it proceeds to elicit these payments, and once these are known, it goes on to elicit beliefs using standard techniques. Of course, this is a rather restrictive setting: in most interesting applications, preferences over states are intrinsic. Besides, in order to elicit the state-dependent payments is quite demanding in terms of the amount of data that is needed.

As we have already mentioned, there is also a large literature within axiomatic decision theory. The various attempts to identify beliefs (without exogenously assuming state-independence) differ in terms of the additional choice domain --- beyond traditional choice data --- that they employ, and of course on the corresponding axioms they impose. For instance, \cite{Fishburn1973} allows for comparison between acts conditional on different events. \cite{KarniSchmeidlerVind1983} and \cite{KarniSchmeidler2016} introduce hypothetical preferences over acts, conditional on exogenously given probabilities over the states. \cite{Kadane1990} allow the agent to compare lotteries at different states. \cite{Dreze1987} and \cite{DrezeRustichini1999} allow for the agent to be able to influence the state realization, in different ways depending on the act she faces. \cite{Lu2019} introduces stochastic choices under different information structures. For a more complete account of this literature, we refer to the reviews of \cite{DrezeRustichini2004}, \cite{Karni2008}, and more recently \cite{Baccelli2017}. Of course, the aim of this whole literature is anyway to establish that beliefs are well-founded and that they can be in principle identified, rather than to suggest concrete methods for actually eliciting said beliefs. In this sense, it is not surprising that using these representation results to actually identify beliefs would be quite a challenge.

Finally, our work is methodologically similar to the one of \cite{Offerman2009}, who first elicit beliefs using standard proper scoring rules, and then design a subsequent test that identifies misreporting due to violations of risk-neutrality and/or presence of probability weighting. 

The paper is structured as follows: Section \ref{S:background} presents the relevant background  concepts. In Section \ref{S:bias identification} we introduce our mechanism and present our results. In Section \ref{S:discussion} we discuss extensions and implementation.

\section{Background}\label{S:background}

\subsection{State-dependent subjective expected utility}

Take a binary state space $\Theta=\{\theta_0,\theta_1\}$. Probability measures over $\Theta$ are identified by the probability they attach to $\theta_1$. Let $\mathcal{L}_X$ be the set of lotteries over a set $X\subseteq\mathbb{R}$ of monetary payoffs. Moreover, let $\mathcal{F}=\mathcal{L}_X^\Theta$ denote the set of acts, with typical element $f$. Consider a (female) agent who maximizes  subjective expected utility (abbrev., SEU). That is, there exist a state-dependent (strictly increasing) Bernoulli utility function $u=(u_0,u_1)$ and a belief $\mu\in(0,1)$, such that her preferences over $\mathcal{F}$ are represented by the function 
\begin{equation}\label{EQ:SDSEU}
\mathbb{E}_\mu(u(f)):=(1-\mu) u_0\bigl(f(\theta_0)\bigr)+ \mu u_1\bigl(f(\theta_1)\bigr), 
\end{equation}
where $u_0(f(\theta_0)):=\langle f(\theta_0),u_0\rangle$ and $u_1(f(\theta_1)):=\langle f(\theta_1),u_1\rangle$ are the (vNM) expected utilities that the lotteries $f(\theta_0)$ and $f(\theta_1)$ yield at states $\theta_0$ and $\theta_1$ respectively. We will say that the SEU representation is state-independent, if $u_0=u_1$.

As it is well-known this representation is not unique. Indeed, for an arbitrary belief $\tilde{\mu}\in(0,1)$, the pair $(\tilde{u},\tilde{\mu})$ is also a subjective expected utility (SEU) representation of the same preferences, if we set 
\begin{equation}
\tilde{u}_0=\frac{1-\mu}{1-\tilde{\mu}}u_0 \ \mbox{ and } \ \tilde{u}_1=\frac{\mu}{\tilde{\mu}}u_1.
\end{equation}
This is because $\mathbb{E}_{\tilde{\mu}}(\tilde{u}(f))=\mathbb{E}_\mu(u(f))$ for every act $f\in\mathcal{F}$. As a result, beliefs cannot be identified with traditional choice data (i.e., preferences over $\mathcal{F}$) alone. Note that identification of beliefs is impossible even if there exists a state-independent SEU representation.\footnote{Recall that a state-independent SEU representation exists whenever the monotonicity axiom of \cite{AnscombeAumann1963} is satisfied.} In order to deal with this fundamental identification problem, different solutions have been proposed in the literature, relying on collecting additional data, well beyond choices over acts. 

One crucial point we should stress is that throughout the paper, we will assume that there exists an actual belief and an actual utility function. Such a pair can be interpreted either as a primitive --- which is actually how we prefer to view it --- or as the parameters that one would obtain by using one of the aforementioned identification results. One way or another, we will say that preferences are state-independent whenever a state-independent SEU representation exists, and moreover it coincides with the \textit{actual} SEU representation. The first part of the previous statement (i.e., existence of a state-independent representation) can be tested with traditional choice data, but the second part (i.e., the state-independent representation being the actual one) needs additional data in order to be tested.


\subsection{Proper scoring rules}

Scoring rules are mechanisms that aim at incentivizing the agent to report her (actual) beliefs truthfully, by rewarding her based on her reported belief and the realized state. Formally, a scoring rule is a function 
\begin{equation*}
\pi:[0,1]\rightarrow\mathcal{F}
\end{equation*}
that takes as an input the reported belief $r\in[0,1]$, and returns as an output the act $\pi_r\in\mathcal{F}$ that the agent receives in return. Formally speaking, a scoring rule is the menu of acts, $\{\pi_r\ | \ 0\leq r\leq1\}$. In this sense, the agent's reported belief is a single point of traditional choice data.

A scoring rule $\pi$ is called binarized whenever it pays in lotteries over two fixed monetary payoffs \citep{HossainOkui2013}, i.e., for every report $r$ and every state $\theta$, the lottery $\pi_r(\theta)$ is distributed over a good payoff $\overline{x}$ and a bad payoff $\underline{x}$. We will refer to the good payoff as the prize, and to the probability of winning the prize as the winning probability. 

A scoring rule is proper if reporting truthfully (uniquely) maximizes the total expected payoff (given the actual beliefs) over the set of all possible reports $r\in[0,1]$. Obviously, for a binarized scoring rule, maximizing the total expected payoff is equivalent to maximizing the total winning probability. 

The appeal of properness is that it claims to identify the agent's beliefs using a single observation of traditional choice data. However --- given the earlier identification problem --- it is not surprising that this cannot be done, unless we impose additional assumptions on the utility functions \citep{KadaneWinkler1988, KarniSafra1995}. To see why this is the case, suppose that some $p\in(0,1)$ has been reported in response to the scoring rule, and observe that there exist infinitely many SEU representations, some with beliefs $\mu<p$ and some others with beliefs $\mu>p$. In principle, we do not know which of these representations corresponds is the actual one. So, in order to identify the actual beliefs we would need to somehow exogenously restrict the set of SEU representations. The way this is typically done is by imposing exogenous assumptions on the utility function. For instance, when a binarized scoring rule is used, it is implicitly assumed that preferences are state-independent. On the other hand, when an arbitrary --- non-binarized --- scoring rule is used, even stronger assumptions are needed, i.e., both state-independence and risk-neutrality are implicitly assumed. The bottom line is that state-independence is always needed if we want to maintain incentive-compatibility. And this has been recognized as perhaps the biggest pitfall of incentivized belief elicitation.


\section{Identifying deviations from actual beliefs}\label{S:bias identification}

From our previous discussion it follows that, if we want to fully identify the agent's beliefs, we will eventually face a fundamental tradeoff. Namely, we will need either to have a rich dataset (going well beyond traditional choice data), or to exogenously assume state-independent preferences. This tradeoff is well-known among theorists, but is often overlooked in practice, where we simply use proper scoring rules without further discussion on the possibility of preferences being state-dependent.  

Here we will take a different approach. We will maintain both the principle of a ``minimal dataset", and we will dispense with the assumption of state-independent preferences. However, in order to be able to accommodate both of these requirements simultaneously, we will relax full identification. In particular, instead of aiming to pin down the agent's actual beliefs, we simply want to learn if the agent has misreported or not. And if she has, we also want to know which direction she has deviated. Notably, we will do all this at the expense of only one additional observation (besides the belief report).

Inspired on a high level by the moral hazard literature \citep{Dreze1987, DrezeRustichini1999,Baccelli2021}, suppose that we can influence the state realization. In particular, assume that there exists some action $\hat{a}$ available to ourselves (viz., the experimenters), which is commonly known to affect the likelihood of $\theta_1$ in a certain direction. Throughout the paper, we will refer to $\hat{a}$ as the \textit{influential action}, and without loss of generality, we will assume that increases the probability of $\theta_1$ to some $\hat{\mu}>\mu$. Notably, we remain agnostic on how much the belief will increase in response to the influential action: all we know is that it will increase. Not picking the influential action $\hat{a}$ means that we stick to the default action $a$ which would leave the agent's beliefs unaffected to $\mu$. Here are a couple of examples of influential actions:

\begin{example}\label{EX:voter}
We are interested in the beliefs of a Democrat about the Democratic candidate winning the upcoming elections. One influential action would be to donate an amount to the Democratic campaign. Another influential action would be to commit some additional votes in a swing state to this candidate (assuming of course that this is a credible commitment). A third influential action would be to start a rumor on social media that the Democratic candidate will increase minimum wages and will decrease taxes. Note that this last action would not involve deception: the influential action is not the realization of the rumor (i.e., the increase of wages or the decrease of taxes) but rather the rumor itself.\footnote{As a disclaimer, we are not recommending experiments that spread fake news. We only use it as an example to illustrate how an influential action functions.} In either case, the agent's subjective probability of the Democratic candidate winning will increase.
\end{example}

\begin{example}\label{EX:investor}
We are interested in an investor's beliefs about a company going bankrupt before the end of the current year. One influential action would be for us to invest money in this company. Another influential action would be to start a rumor that the company is about to file for new patent.\footnote{Here the same comment (regarding deception) applies as in the previous example.} In both cases, it is reasonable to assume that the investor's subjective belief of bankruptcy will go down. 
\end{example}

\begin{example}\label{EX:driver}
We are interested in the beliefs of a young economist about her paper being published in a top journal. One influential action would be to put a good word with a friendly editor. Another influential action would be to commit that a prominent economist will carefully read the manuscript and provide comments before the paper is submitted. In both cases, the subjective probability the author assigns to the paper being accepted will go up.
\end{example}

\begin{remark}\label{REM:choice of influential action}
It is really important to choose an influential action which does not affect the agent directly, besides the effect that it has on the state space. For instance in Example \ref{EX:driver}, we should not try to influence the editor if the author of the paper has ethical issues with lobbying. This is because in such case, the influential action would distort  not only the agent's beliefs, but also her utilities.
\end{remark}

Notice that in our case, it is us (viz., the experimenters) who control the influential action, as opposed to the moral hazard literature where the action is controlled by the agent. As a result, we can construct lotteries over the product space $X\times \{a,\hat{a}\}$. These are not usual lotteries that pay only in monetary payoffs. Instead, an outcome of such a lottery would be a pair of a monetary payoff (which affects the agent directly) and an action (which affects the agent indirectly via the uncertainty on $\Theta$).

For the two monetary outcomes, $\underline{x}$ and $\overline{x}$, define the lotteries:
\begin{eqnarray}
A&:=&\biggl(\frac{1}{2}\times(\underline{x},a),\frac{1}{2}\times(\overline{x},\hat{a})\biggr),\nonumber\\
& & \\
B&:=&\biggl(\frac{1}{2}\times(\overline{x},a),\frac{1}{2}\times(\underline{x},\hat{a})\biggr). \nonumber
\end{eqnarray}
Intuitively, in the context of Example \ref{EX:voter}, suppose that the good payoff is $\overline{x}=\text{\$10k}$, while the influential action is a donation of $\hat{a}=\text{\$10k}$. Then, $A$ can be seen as a ``risky option" for the agent, in the sense that either \$20k will be paid out in total (\$10k to herself and \$10k to the campaign), or no money at all will be paid out. On the other hand, $B$ can be seen as a ``hedging option" for the agent, in the sense that \$10k will be paid out anyway, either to the agent herself or to campaign.

\begin{remark}\label{REM:moral hazard}
By having a choice between $A$ and $B$, the agent cannot influence the state realization. This is because regardless which of the two lotteries is chosen, the influential action and the default action will both occur with probability one half. This is a major difference with the moral hazard literature, which relies on the agent being able to affect the state. 
\end{remark}

Then, the following results use the agent's revealed preferences over the pair of lotteries to identify misreporting, viz., $A$ is chosen (resp., $B$ is chosen) if the reported belief is above (resp., below) the actual beliefs.

\begin{theorem}\label{T:binarized}
Let $\mu$ be the agent's actual beliefs, and $p$ be her reported beliefs in response to a proper binarized scoring rule. Then, we have $p> \mu$ (resp., $p<\mu$), if and only if, $A\succ B$ (resp., $A\prec B$).
\end{theorem}

\begin{proof}[\textsc{Proof}]
First, denote by $\pi_r^k:=\pi_r(\theta_k)(\overline{x})$ the winning probability (at state $\theta_k$) when the report $r$ is submitted. By properness of $\pi$, we have 
\begin{eqnarray*}
(1-r)\pi_r^0+r\pi_r^1&>&(1-r)\pi_p^0+r\pi_p^1,\\
(1-p)\pi_r^0+p\pi_r^1&<&(1-p)\pi_p^0+p\pi_p^1.
\end{eqnarray*}
For any $r<p$, we have $\pi_r^0<\pi_p^0$ and $\pi_r^1>\pi_p^1$, and therefore
\begin{equation}\label{EQ:proof properness 1}
\frac{r}{1-r}<\frac{\pi_p^0-\pi_r^0}{\pi_r^1-\pi_p^1}<\frac{p}{1-p},
\end{equation}
whereas for any $r>p$, we have $\pi_r^0>\pi_p^0$ and $\pi_r^1<\pi_p^1$, and therefore we obtain
\begin{equation}\label{EQ:proof properness 2}
\frac{p}{1-p}<\frac{\pi_p^0-\pi_r^0}{\pi_r^1-\pi_p^1}<\frac{r}{1-r}.
\end{equation}
Taking side limits of in (\ref{EQ:proof properness 1}) and (\ref{EQ:proof properness 2}) as $r$ approaches $p$ from the below and above respectively, yields
\begin{equation}\label{EQ:proof properness 3}
\lim_{r\uparrow p}\frac{\pi_p^0-\pi_r^0}{\pi_r^1-\pi_p^1}= \lim_{r\downarrow p}\frac{\pi_p^0-\pi_r^0}{\pi_r^1-\pi_p^1}=\frac{p}{1-p}.
\end{equation}

\vspace{0.5\baselineskip} \noindent Now, given the actual belief $\mu$, the expected utility from the report $r$ is equal to
\begin{equation}
\mathbb{E}_\mu\bigl(u(\pi_r)\bigr)=(1-\mu)\Bigl(\pi_r^0 u_0(\overline{x})+\bigl(1-\pi_r^0\bigr) u_0(\underline{x})\Bigr)+\mu\Bigl(\pi_r^1 u_1(\overline{x})+\bigl(1-\pi_r^1\bigr)u_1(\underline{x})\Bigr).
\end{equation}
By $p$ being actually reported, it follows that $\mathbb{E}_\mu\bigl(u(\pi_p)\bigr)\geq\mathbb{E}_\mu\bigl(u(\pi_r)\bigr)$ for every $r\neq p$. This means that whenever it is the case that $r<p$, we have
\begin{equation}\label{EQ:proof properness 4}
\frac{\mu}{1-\mu}\cdot\frac{u_1(\overline{x})-u_1(\underline{x})}{u_0(\overline{x})-u_0(\underline{x})}\geq \frac{\pi_p^0-\pi_r^0}{\pi_r^1-\pi_p^1},
\end{equation}
while whenever it is the case that $r>p$, we obtain
\begin{equation}\label{EQ:proof properness 5}
\frac{\mu}{1-\mu}\cdot\frac{u_1(\overline{x})-u_1(\underline{x})}{u_0(\overline{x})-u_0(\underline{x})}\leq \frac{\pi_p^0-\pi_r^0}{\pi_r^1-\pi_p^1}.
\end{equation}
Hence, if we take the side limits in (\ref{EQ:proof properness 4}) and (\ref{EQ:proof properness 5}) as $r$ approaches $p$ from below and above respectively, and we use Equation (\ref{EQ:proof properness 3}), it will follow that
\begin{equation}
\frac{p}{1-p}=\frac{\mu}{1-\mu}\cdot\frac{u_1(\overline{x})-u_1(\underline{x})}{u_0(\overline{x})-u_0(\underline{x})}.
\end{equation}
Obviously, this directly implies the equivalence
\begin{equation}\label{EQ:proof theorem 1}
p\geq\mu \ \Leftrightarrow \ \frac{u_1(\overline{x})-u_1(\underline{x})}{u_0(\overline{x})-u_0(\underline{x})}\geq1,
\end{equation}
with the first inequality being strict, if and only if, the second one is strict.

\vspace{0.5\baselineskip} \noindent Now, let $\hat{\mu}>\mu$ be the unobserved probability that the agent attaches to $\theta_1$ if $\hat{a}$ is chosen. Then, the following equivalences hold:
\begin{eqnarray*}
A\succeq B&\Leftrightarrow&\frac{1}{2}\mathbb{E}_{\mu}\bigl(u(\underline{x})\bigr)+\frac{1}{2}\mathbb{E}_{\hat{\mu}}\bigl(u(\overline{x})\bigr)\geq \frac{1}{2}\mathbb{E}_{\mu}\bigl(u(\overline{x})\bigr)+\frac{1}{2}\mathbb{E}_{\hat{\mu}}\bigl(u(\underline{x})\bigr)\\
& & \\
&\Leftrightarrow& \mathbb{E}_{\hat{\mu}}\bigl((u(\overline{x})-u(\underline{x}\bigr))\geq\mathbb{E}_\mu\bigl((u(\overline{x})-u(\underline{x})\bigr)\\
& & \\
&\Leftrightarrow&(\hat{\mu}-\mu)\bigl(u_1(\overline{x})-u_1(\underline{x})\bigr)\geq (\hat{\mu}-\mu)\bigl(u_0(\overline{x})-u_0(\underline{x})\bigr)\\
& & \\
&\Leftrightarrow& \frac{u_1(\overline{x})-u_1(\underline{x})}{u_0(\overline{x})-u_0(\underline{x})}\geq1,
\end{eqnarray*}
with the last inequality being strict, if and only if the preference relation is strict. Combining this last equivalence with (\ref{EQ:proof theorem 1}) completes the proof.
\end{proof}

\begin{remark}
In case the influential action is known to decrease --- rather than increase --- the subjective probability of $\theta_1$, the previous result still stands verbatim with the preference ordering reversed, i.e., $p>\mu$, if and only if, $A\prec B$.
\end{remark}

Note that the previous results crucially rely on us being able to set the probabilities in each of the two lotteries to exactly fifty-fifty. Let us explain why fifty-fifty probabilities are so crucial. First, note that the agent will over-report her belief of $\theta_1$, if and only if, the utility function is (locally) supermodular, i.e., formally speaking, the difference $u_1-u_0$ is increasing as we move from $\underline{x}$ to $\overline{x}$ (see (\ref{EQ:proof theorem 1})).\footnote{The same condition is obtained by \cite{KadaneWinkler1988} and \cite{JaffrayKarni1999} for matching probabilities. This implies that our method works verbatim if we replace binarized scoring rules with matching probabilities. In fact, the same is true for other elicitation methods (see Section \ref{S:other elicitation methods}).} Then, we go on to show that this supermodularity condition is characterized by the preferences over these exact fifty-fifty lotteries.\footnote{Interestingly, a similar condition is used by \cite{Francetich2013} in his characterization result of supermodular vNM EU functions.} Let us illustrate why this is the case. Suppose that the payments are $\underline{x}=0$ and $\overline{x}=1$. Since the influential action leads to an increased probability $\hat{\mu}>\mu$, it will be the case that $u_1-u_0$ is increasing, if and only if, $\mathbb{E}_{\hat{\mu}}(u(\cdot))-\mathbb{E}_\mu(u(\cdot))$ is increasing.  
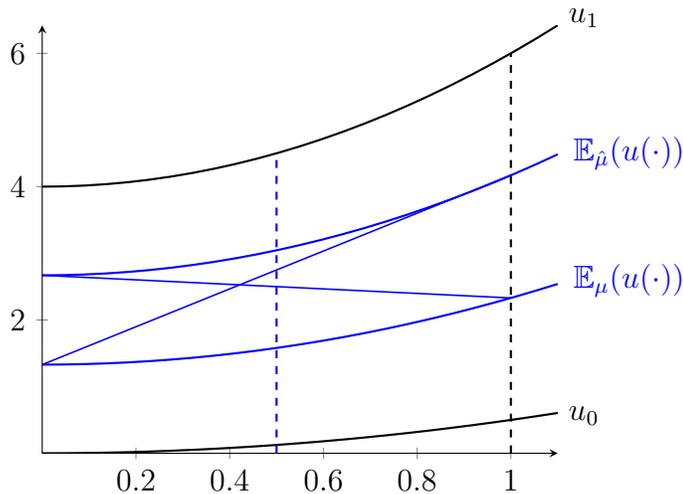
\begin{figure}[h!]
\centering
\begin{tikzpicture}
   \begin{axis}[my style]
               \addplot[domain=0:1.1,line width=0.8] {0.5*x^2};
               \addplot[domain=0:1.1,line width=0.8] {4+2*x^2};
             \addplot[domain=0:2,line width=0.8,blue,dashed]   coordinates { (0.5,0) (0.5,4.5) };
             \addplot[domain=0:2,line width=0.8,dashed]   coordinates { (1,0) (1,6) };
               \addplot[domain=0:1.1,line width=0.8,blue] {1.33+x^2};
               \addplot[domain=0:1.1,line width=0.8,blue] {2.67+1.5*x^2};
             \addplot[domain=0:1,blue,line width=0.6]   coordinates { (0,1.33) (1,4.17) };
             \addplot[domain=0:1,blue,line width=0.6]   coordinates { (0,2.67) (1,2.33) };
\end{axis}
          \draw (7.2,5.8) node {$u_1$};
         \draw[blue] (7.8,4) node {$\mathbb{E}_{\hat{\mu}}(u(\cdot))$};
        \draw[blue] (7.8,2.3) node {$\mathbb{E}_{\mu}(u(\cdot))$};
          \draw (7.2,0.5) node {$u_0$};
\end{tikzpicture}
\caption{Supermodularity is characterized by preferences over $\{A,B\}$. }
\label{FIG:Supermodularity}
\end{figure}
Then, take the two straight lines, one that connects the graph of $\mathbb{E}_{\hat{\mu}}(u(\cdot))$ evaluated at $1$ to the graph of $\mathbb{E}_\mu(u(\cdot))$ evaluated at $0$, and one that connects the graph of $\mathbb{E}_\mu(u(\cdot))$ evaluated at $1$ to the graph of $\mathbb{E}_{\hat{\mu}}(u(\cdot))$ evaluated at $0$. Finally, observe that the difference $\mathbb{E}_{\hat{\mu}}(u(\cdot))-\mathbb{E}_\mu(u(\cdot))$ in increasing, if and only if, these two lines intersect to left of $1/2$. But then again, the two lines intersect to the left of $1/2$, if and only if, $A$ is preferred to $B$. 

An alternative approach would have been to induce fifty-fifty probabilities via an information structure that yields two signals, each occurring with probability a half. In particular, suppose that we can construct an experiment that yields either signal $s_0$ (which is known to increase the probability of $\theta_0$) or signal $s_1$ (which is known to increase the probability $\theta_1$). Then, we would be asking the subject whether she prefers to be paid the good payoff $\overline{x}$ when $s_0$ is realized and the bad payoff $\underline{x}$ when $s_1$ is realized, or vice versa \citep[similarly to][]{Lu2019}. However, this alternative mechanism would rely on two very strong assumptions. First, we would need to make sure that we can design such an experiment. However, this would be practically impossible unless we knew the prior $\mu$, which of course we do not know. Second, if we were hypothetically able to design such an experiment, we would need to know that the agent updates in a Bayesian manner. This we do not know either. So overall, as theoretically appealing as this alternative mechanism may look, the implementation would be rather difficult.

\section{Discussion}\label{S:discussion}

\subsection{Non-binarized scoring rules}\label{S:non-binarized scoring rules}

The reason binarized scoring rules are appealing is because they do not require any assumption --- besides state-independence --- in order to guarantee truth-telling. This is in contrast to arbitrary proper scoring rules (e.g., the commonly-used quadratic scoring rule) which need to assume risk-neutrality --- on top of state-independence --- in order to retain incentive-compatibility. Of course, we should note that there is a debate on the tradeoff between incentive-compatibility and not needing reduce compound lotteries, e.g., see \cite{Selten1999}, \cite{Harrison2013}, \cite{Harrison2014}, \cite{Harrison2015}, just to mention a few. Although we personally find the overall evidence to favor binarized scoring rules, it is not our aim to participate in this discussion. Instead we ask the following question: if we assume risk-neutrality at both states, can we identify misreporting due to state-independence? The answer is affirmative: Using Theorem \ref{T:binarized} for any two payments $\underline{x}$ and $\overline{x}$ identifies at which of the two states the marginal utility is greater, which in turn reveals misreporting. In this sense, our method is not restricted to the binarized case.

\subsection{Beyond scoring rules}\label{S:other elicitation methods}

As we have already mentioned, our methodology holds verbatim if we replace binarized scoring rules with any other binarized elicitation task, such as matching probabilities or clock auctions. The reason is that the belief elicitation task is essentially independent of our additional task that identifies misreporting. For instance, similarly to our analysis of scoring rules, \cite{KadaneWinkler1988} and \cite{JaffrayKarni1999} show that matching probabilities will induce misreporting in favor of $\theta_1$, if and only if, the difference $u_1-u_0$ is increasing. Then, the choice between our lotteries $A$ and $B$ characterizes the monotonicity of $u_1-u_0$. Hence, our method will tell us whether the reported belief deviates from the actual one, and if so, in which direction.

\subsection{Hedging}\label{S:hedging}

A well-known concern regarding incentivized belief elicitation is the possibility of hedging \citep{Blanco2010}. In general terms this means that the agent can choose a optimal strategy for the grand decision problem which includes both the elicitation task and some other task, which does not induce truth-telling in the elicitation task. The most common manifestation of the problem is due to non-risk-neutral risk preferences. In our case, this is not really a problem, as long as we pay randomly for one of the two tasks, i.e., either the binarized scoring rule or the choice from $\{A,B\}$. Crucially, the chosen lottery in the second task will be realized only if the second task has been drawn to be compensated. This is done so that beliefs in the first task are not distorted in anticipation of the possibility that the influential action will be drawn in the second task.


\begin{thebibliography}{}


\bibitem[\protect\citeauthoryear{Anscombe and Aumann}{1963}]{AnscombeAumann1963} \textsc{Anscombe, F.J. \& Aumann, R.J.} (1963). A definition of subjective probability. \textit{Annals of Mathematical Statistics} 34, 199--205.

\bibitem[\protect\citeauthoryear{Baccelli}{2017}]{Baccelli2017} \textsc{Baccelli, J.} (2017). Do bets reveal beliefs? A unified perspective on state-dependent utility issues. \textit{Synthese} 194, 3393--3419.


\bibitem[\protect\citeauthoryear{Baccelli}{2021}]{Baccelli2021} -------- (2021). Moral hazard, the Savage framework, and state‐dependent utility. \textit{Erkenntnis} 86, 367--387.

\bibitem[\protect\citeauthoryear{Baillon et al.}{2018}]{Baillon2018} \textsc{Baillon, A., Huang, Z., Selim, A. \& Wakker, P.} (2018). Measuring ambiguity attitudes for all (natural) events. \textit{Econometrica} 86, 1839--1858.







\bibitem[\protect\citeauthoryear{Blanco et al.}{2010}]{Blanco2010} \textsc{Blanco, M., Engelmann, D., Koch, A. \& Normann, H.T.} (2010). Belief elicitation in experiments: is there a hedging problem? \textit{Experimental Economics} 13, 412--438.

\bibitem[\protect\citeauthoryear{Brier}{1950}]{Brier1950} \textsc{Brier, G.} (1950). Verification of forecasts expressed in terms of probability. \textit{Monthly Weather Review} 78, 1--3.


\bibitem[\protect\citeauthoryear{Bullock et al.}{2015}]{Bullock2015} \textsc{Bullock, J., Gerber, A., Hill, S. \& Huber, G.} (2015). Partisan bias in factual beliefs about politics. \textit{Quarterly Journal of Political Science} 10, 519--578.




\bibitem[\protect\citeauthoryear{De Finetti}{1974}]{DeFinetti1974} \textsc{De Finetti, B.} (1974). \textit{Theory of probability}. Vol. 1, Wiley.

\bibitem[\protect\citeauthoryear{Dr\`{e}ze}{1987}]{Dreze1987} \textsc{Dr\`{e}ze, J.} (1987). Decision theory with moral hazard and state-dependent preferences. \textit{Essays on Economic Decisions under Uncertainty} 23--89.

\bibitem[\protect\citeauthoryear{Dr\`{e}ze and Rustichini}{1999}]{DrezeRustichini1999} \textsc{Dr\`{e}ze, J. \& Rustichini, A.} (1999). Moral hazard and conditional preferences. \textit{Journal of Mathematical Economics} 31, 159--181.

\bibitem[\protect\citeauthoryear{Dr\`{e}ze and Rustichini}{2004}]{DrezeRustichini2004} -------- (2004). State-dependent utility and decision theory. \textit{Handbook of Utility Theory}, Ch. 8, 839--892.

\bibitem[\protect\citeauthoryear{Ducharme and Donnell}{1973}]{DucharmeDonnell1973} \textsc{Ducharme, W. \& Donnell, M.} (1973). Intrasubject comparison of four response modes  for subjective probability assessment. \textit{Organizational Behavior and Human Performance} 10, 108--117.


\bibitem[\protect\citeauthoryear{Ewers and Zimmermann}{2015}]{EwersZimmermann2015} \textsc{Ewers, M. \& Zimmermann, F.} (2015). Image and misreporting. \textit{Journal of teh European Economic Association} 13, 363--380.

\bibitem[\protect\citeauthoryear{Fishburn}{1973}]{Fishburn1973} \textsc{Fishburn, P.} (1973). A mixture-set axiomatization of conditional subjective expected utility. \textit{Econometrica} 41, 1--25.

\bibitem[\protect\citeauthoryear{Francetich}{2013}]{Francetich2013} \textsc{Francetich, A.} (2013). Notes on supermodularity and increasing differences in expected utility. \textit{Economics Letters} 121, 206--209.



\bibitem[\protect\citeauthoryear{Good}{1952}]{Good1952} \textsc{Good, I.J.} (1952). Rational decisions. \textit{Journal of the Royal Statistical Society, Series B}, 14, 107--114.

\bibitem[\protect\citeauthoryear{Hannon and de Ridder}{2021}]{HannonRidder2021} \textsc{Hannon, M. \& de Ridder, J.} (2021). The point of political belief. \textit{Routledge Handbook of Political Epistemology} (forthcoming).

\bibitem[\protect\citeauthoryear{Harrison et al.}{2013}]{Harrison2013} \textsc{Harrison, G., Mart\`{i}nez-Correa, J. \& Swarthout, J.T.} (2013).  Inducing risk neutral preferences with binary lotteries: a reconsideration. \textit{Journal of Economic Behavior and Organization} 94, 145--159.

\bibitem[\protect\citeauthoryear{Harrison et al.}{2014}]{Harrison2014} -------- (2014). Eliciting subjective probabilities with binary lotteries. \textit{Journal of Economic Behavior and Organization} 101, 128--140.

\bibitem[\protect\citeauthoryear{Harrison et al.}{2015}]{Harrison2015} \textsc{Harrison, G., Mart\`{i}nez-Correa, J., Swarthout, J.T. \& Ulm, E.} (2015).  Eliciting subjective probability distributions with binary lotteries. \textit{Economics Letters} 127, 68--71.


\bibitem[\protect\citeauthoryear{Hossain and Okui}{2013}]{HossainOkui2013} \textsc{Hossain, T. \& Okui, R.} (2013). The binarized scoring rule. \textit{Review of Economic Studies} 80, 984--1001.

\bibitem[\protect\citeauthoryear{Jaffray and Karni}{1999}]{JaffrayKarni1999} \textsc{Jaffray, J.Y. \& Karni, E.} (1999). Elicitation of subjective probabilities when the initial endowment is unobservable. {\it Journal of Risk and Uncertainty} 8, 5--20.

\bibitem[\protect\citeauthoryear{Kadane and Winkler}{1988}]{KadaneWinkler1988} \textsc{Kadane, J. \& Winkler, R.} (1988). Separating probability elicitation from utilities. {\it Journal of the American Statistical Association} 83, 357--363.


\bibitem[\protect\citeauthoryear{Karni}{1992}]{Karni1992} \textsc{Karni, E.} (1992). Subjective probabilities and utilities with event-dependent preferences. {\it Journal of Risk and Uncertainty} 5, 107--125.

\bibitem[\protect\citeauthoryear{Karni}{1993}]{Karni1993} -------- (1993). A definition of subjective probabilities with state-dependent preferences. {\it Econometrica} 61, 187--198.

\bibitem[\protect\citeauthoryear{Karni}{1999}]{Karni1999} -------- (1999). Elicitation of subjective probabilities when preferences are state-dependent. {\it International Economic Review} 40, 479--486.

\bibitem[\protect\citeauthoryear{Karni}{2008}]{Karni2008} -------- (2008). State-dependent utility. {\it Handbook of Rational and Social Choice}, 223--238.

\bibitem[\protect\citeauthoryear{Karni}{2009}]{Karni2009} -------- (2009). A mechanism for eliciting probabilities. {\it Econometrica} 77, 603--606.

\bibitem[\protect\citeauthoryear{Karni and Safra}{1995}]{KarniSafra1995} \textsc{Karni, E. \& Safra, Z.} (1995). The impossibility of experimental elicitation of subjective probabilities. {\it Theory and Decision} 38, 313--320.

\bibitem[\protect\citeauthoryear{Karni and Schmeidler}{2016}]{KarniSchmeidler2016} \textsc{Karni, E. \& Schmeidler, D.} (2008). An expected utility theory for state-dependent preferences. {\it Theory and Decision} 81, 467--478.

\bibitem[\protect\citeauthoryear{Karni et al.}{1983}]{KarniSchmeidlerVind1983} \textsc{Karni, E., Schmeidler, D. \& Vind, K.} (1983). On state-dependent preferences and subjective probabilities. {\it Econometrica} 51, 1021--1031.


\bibitem[\protect\citeauthoryear{Lu}{2019}]{Lu2019} \textsc{Lu, J.} (2019). Bayesian identification: a theory for state-dependent utilities. {\it American Economic Review} 109, 3192--3228.

 \bibitem[\protect\citeauthoryear{Manski}{2004}]{Manski2004} \textsc{Manski, C.F.} (2004). Measuring expectations. \textit{Econometrica} 72, 1329--1376.




\bibitem[\protect\citeauthoryear{Offerman et al.}{2009}]{Offerman2009} \textsc{Offerman, T., Sonnemans, J., van de Kuilen, G. \& Wakker, P.} (2009). A truth serum for non-Bayesians: correcting proper scoring rules for risk attitudes. \textit{Review of Economic Studies} 76, 1461--1489.

\bibitem[\protect\citeauthoryear{Savage}{1954}]{Savage1954}
\textsc{Savage, L.} (1954). {\it The foundations of statistics.} Wiley, NY: Dover Publications.

\bibitem[\protect\citeauthoryear{Savage}{1971}]{Savage1971} -------- (1971). Elicitation of personal probabilities and expectations. \textit{Journal of the American Statistical Association} 66, 783--801.

\bibitem[\protect\citeauthoryear{Schervish et al.}{1990}]{Kadane1990} \textsc{Schervish, M., Seidenfeld, T. \& Kadane, J.} (1990). State-dependent utilities. {\it Journal of the American Statistical Association} 85, 840--847.

\bibitem[\protect\citeauthoryear{Selten et al.}{1999}]{Selten1999} \textsc{Selten, R., Sadrieh, A. \& Abbink, K.} (1999). Money does not induce risk neutral behavior, but binary lotteries do even worse. \textit{Theory and Decision} 46, 211--249.

\bibitem[\protect\citeauthoryear{Svenson}{1981}]{Svenson1981} \textsc{Svenson, O.} (1981). Are we all less risky and more skillfull than our fellow drivers?. \textit{Acta Psychologica} 47, 143--148.


\bibitem[\protect\citeauthoryear{Thaler}{2021}]{Thaler2021} \textsc{Thaler, M.} (2021). The supply of motivated beliefs. \textit{Working Paper}.

\bibitem[\protect\citeauthoryear{Tsakas}{2019}]{Tsakas2019} \textsc{Tsakas, E.} (2019). Obvious belief elicitation. \textit{Games and Economic Behavior} 118, 374--381.


\bibitem[\protect\citeauthoryear{Wakker}{1993}]{Wakker1993} \textsc{Wakker, P.} (1993). Unbounded utility for Savage's ``foundations of statistics", and other models. {\it Mathematics of Operations Research} 18, 446--485.




\end{thebibliography}
\end{document}